\documentclass{article}
\usepackage{geometry}
\usepackage{amsthm}
\usepackage{amsfonts}
\usepackage{amssymb}
\usepackage{amsmath}

\usepackage{multirow}

\usepackage{mathtools}

\usepackage[colorinlistoftodos]{todonotes}
\usepackage{fancyhdr}

\usepackage{cite}

\usepackage{hyperref}
\usepackage{url}
\usepackage{tabularx}

\newcommand{\tbf}[1]{\textbf{#1}}

\DeclarePairedDelimiter\abs{\lvert}{\rvert}%
\DeclarePairedDelimiter\norm{\lVert}{\rVert}%

\makeatletter
\let\oldabs\abs
\def\abs{\@ifstar{\oldabs}{\oldabs*}}

\let\oldnorm\norm
\def\norm{\@ifstar{\oldnorm}{\oldnorm*}}
\makeatother

\DeclarePairedDelimiter{\ceil}{\lceil}{\rceil}%
\DeclarePairedDelimiter{\floor}{\lfloor}{\rfloor}%

\makeatletter
\let\oldceil\ceil
\def\ceil{\@ifstar{\oldceil}{\oldceil*}}

\makeatletter
\let\oldfloor\floor
\def\floor{\@ifstar{\oldfloor}{\oldfloor*}}
\DeclarePairedDelimiter{\tuple}{\langle}{\rangle}%

\makeatletter
\let\oldtuple\tuple
\def\tuple{\@ifstar{\oldtuple}{\oldtuple*}}

\theoremstyle{plain}
\newtheorem{theorem}{Theorem}
\newtheorem{definition}[theorem]{Definition}
\newtheorem{claim}[theorem]{Claim}
\newtheorem{observation}[theorem]{Observation}
\newtheorem{lemma}[theorem]{Lemma}

\newif\iffullver \fullverfalse
\newif\ifoldversions \oldversionsfalse

\DeclareMathOperator{\Write}{\mathtt{write}}
\DeclareMathOperator{\Writes}{\mathtt{write}\text{s}}
\DeclareMathOperator{\Read}{\mathtt{read}}
\DeclareMathOperator{\Reads}{\mathtt{read}\text{s}}

\newcommand{\mytext}[1]{\text{\textit{#1}}}
\newcommand{\data}{\mytext{data}}
\newcommand{\meta}{meta-data}
\newcommand{\ie}{i.e., }
\newcommand{\eg}{e.g., }
\newcommand{\bobj}{o}
\newcommand{\bo}[1]{\bobj_{#1}}
\renewcommand{\wr}[1]{\Write(#1)}
\newcommand{\V}{\mathbb{V}}

\newcommand{\C}{\mathbb{B}}
\newcommand{\W}{\mathbb{W}}
\newcommand{\N}{\mathbb{N}}
\newcommand{\op}{\mathtt{a}}
\newcommand{\update}{\mathtt{update}}
\newcommand{\updates}{\mathtt{update}\text{s}}
\newcommand{\get}{\mathtt{get}}
\renewcommand{\gets}{\mathtt{get}\text{s}}
\newcommand{\step}{\op_p}

\newcommand{\alg}{\mathcal{A}}
\newcommand{\commonWrite}{$\tau$-common write}
\newcommand{\disint}{$\tau$-disintegrated storage}
\newcommand{\values}[1]{\mytext{values}_{p} \left(#1 \right)}
\newcommand{\writes}[1]{\mytext{writes}_{p} \left( #1 \right)}
\newcommand{\sharedStorage}{shared storage}
\newcommand{\twor}[1]{\multirow{2}{*}{#1}}
\newcommand{\tomath}[1]{\relax\ifmmode#1\else$#1$\fi}
\newcommand{\resCommonValueInvisible}{\tomath{\tau + (\tau-1) \ceil{\mneedceil}}}
\newcommand{\resCommonValueVisible}{\tomath{\tau + (\tau-1) \cdot \min \left( \ceil{\mneedceil} \, , \, R \right) }}
\newcommand{\resCommonWriteInvisible}{\tomath{\tau \cdot 2^D}}
\newcommand{\resCommonWriteVisible}{\tomath{\tau + (\tau-1) \cdot \min \left(2^D - 1 \, , \, R \right)}}
\newcommand{\object}{object}
\newcommand{\objects}{objects}

\newcommand{\Objects}{Objects}
\newcommand{\indist}{\approx_w}
\newcommand{\mneedceil}{\frac{2^D-1}{L}}
\newcommand{\mnoneedceil}{-1}
\newcommand{\mdef}{\ceil{\mneedceil} \mnoneedceil}
\newcommand{\commonValueN}{\min \left( \ceil{\mneedceil} \, , \, R \right) - 1}
\newcommand{\commonWriteN}{\min \left(2^D - 1 \, , \, R \right) - 1}
\newcommand{\slabelsnoarg}{\mytext{S--labels}}
\newcommand{\llabelsnoarg}{\mytext{L--labels}}
\newcommand{\alllabelsnoarg}{\mytext{All--labels}}
\newcommand{\slabels}[1]{\slabelsnoarg{} \left(#1\right)}
\newcommand{\llabels}[1]{\llabelsnoarg{}_p\left(#1\right)}
\newcommand{\alllabels}[1]{\alllabelsnoarg{}_p\left(#1\right)}

\newcommand{\mylabel}[1]{\left\langle #1 \right\rangle}
\newcommand{\Label}[1]{\mytext{Labels}(#1)}
\newcommand{\sharedlabel}[1]{\Label{#1}}

\newcommand{\finalof}[1]{t_{#1}}
\newcommand{\invariantValues}[1]{\varphi \left( #1 \right)}
\newcommand{\invariantWrites}[1]{\psi \left( #1 \right)}
\newcommand{\pdata}[2][p]{#1.\data \left(#2\right)}
\newcommand{\odata}[1]{\pdata[o]{#1}}
\newcommand{\suffix}[2]{$#1 \setminus #2$}
\newcommand{\permanent}[3]{\tuple{#1, \; #3, \; #2} \text{-permanent}}
\newcommand{\rfinal}[1]{#1, \finalof{#1}}

\newcommand{\constant}[3][\tau-1]{\tuple{#1, \; #3, \; #2} \text{-constant}}

\newcolumntype{Y}{>{\centering\arraybackslash}X}
\title{Integrated Bounds for Disintegrated Storage}

\usepackage{authblk}
\newcommand{\myparagraph}[1]{\paragraph{#1}\mbox{}\\ \mbox{}\\}
\author[1]{Alon Berger}

\author[1]{Idit Keidar}

\author[1,2]{Alexander Spiegelman}

\affil[1]{Viterbi Department of Electrical Engineering, Technion, Haifa, Israel}
\affil[2]{VMware Research, Israel}

\date{}

\begin{document}
	\maketitle
	
\abstract{
	We point out a somewhat surprising similarity between non-authenticated Byzantine storage, coded storage, and certain emulations of shared registers from smaller ones.
	A common characteristic in all of these is the inability of reads to safely return a value obtained in a single atomic access to shared storage.
	We collectively refer to such systems as \emph{disintegrated storage}, and show integrated space lower bounds for asynchronous regular wait-free emulations in all of them.
	In a nutshell, if readers are invisible, then the storage cost of such systems is inherently exponential in the size of written values; otherwise, it is at least linear in the number of readers.
	Our bounds are asymptotically tight to known algorithms, and thus justify their high costs.
}

\section{Introduction}
\label{sec:introduction}

\subsection{Space bounds for encoded, multi-register, and Byzantine storage}

In many data sharing solutions, information needs to be read from multiple sources in order for a single value to be reconstructed.
One such example is coded storage where multiple storage blocks need to be obtained in order to recover a single value that can be returned to the application~\cite{codingCachin2006optimal, codingLynchCadambe2014coded, codingRashid, codingGoodson2004efficient, spaceBounds, peterson1983concurrent, dobre2013powerstore, androulaki2014erasure}.
Another example arises in shared memory systems, where the granularity of atomic memory operations (such as load and store) is limited to a single word (e.g., 64 bits) and one wishes to atomically read and write larger values~\cite{peterson1983concurrent}.
A third example is replicating data to overcome Byzantine faults (without authentication) or data corruption, where a reader expects to obtain the same block from multiple servers in order to validate it~\cite{faultyMemoryJayanti1998, faultyMemorybyzantineDiscPaxos, abraham2007wait}.

We refer to such systems collectively as \emph{disintegrated storage} systems.
We show that such a need to read data in multiple storage accesses inherently entails high storage costs: exponential in the data size if reads do not modify the storage, and otherwise linear in the number of concurrent reads.
This stands in contrast to systems that use non-Byzantine replication, such as ABD~\cite{ABD}, where, although \meta{} (\eg timestamps) is read from several sources, the recovered value need only be read from a single source.

\subsection{Our results}

We consider a standard shared storage model (see Section~\ref{sec:model}).
We refer to shared storage locations (representing memory words, disks, servers, etc.) as \emph{\objects}.
To strengthen our lower bounds, we assume that \objects{} are responsive, \ie do not fail; the results hold a fortiori if \objects{} can also be unresponsive~\cite{faultyMemoryJayanti1998}.
\Objects{} support general read-modify-write operations by asynchronous \emph{processes}.
We study \emph{wait-free} emulations of a shared \emph{regular register}~\cite{lamportRegular}.

Section~\ref{sec:disintegrated_storage} formally defines \emph{disintegrated storage}.
We use a notion of \emph{blocks}, which are parts of a value kept in storage -- code blocks, segments of a longer-than-word value, or full copies of a replicated value.
A key assumption we make is that each block in the shared storage pertains to a single write operation; a similar assumption was made in previous studies~\cite{spaceBounds, cadambe2016podc}.
The disintegration property then stipulates that a reader must obtain some number $\tau > 1$ of blocks pertaining to a value $v$ before returning $v$.
For example, $\tau$ blocks are needed in $\tau$-out-of-$n$ coded storage, whereas $\tau = f+1$ in $f$-tolerant Byzantine replication.
To strengthen our results, we allow the storage to hold unbounded \meta{} (\eg timestamps), and count only the storage cost for blocks.
Note that the need to obtain $\tau$ blocks implies that \meta{} cannot be used instead of actual data.

In Section~\ref{sec:commonValue} we give general lower bounds that apply to all types of disintegrated storage~-- replicated, coded, and multi-register.
We first consider \emph{invisible reads}, which do not modify the shared storage.
This is a common paradigm in storage systems and often essential where readers outnumber writers and have different permissions.
In this case, even with one reader and one writer, the storage size can be exponential; specifically, if value sizes are $D$ (taken from a domain of size $2^D$), then we show a lower bound of \resCommonValueInvisible{} blocks, where $L$ is the number of blocks in a reader's local storage.
That is, either the local storage of the reader or the shared storage is exponential.

Section~\ref{sec:commonWrite} studies a more restrictive flavor of disintegrated storage, called \emph{\commonWrite}, where a reader needs to obtain $\tau$ blocks produced by the \emph{same} $\wr{v}$ operation in order to return $v$.
In other words, if the reader obtains blocks that originate from two different $\Writes$ of the same value, then it cannot recognize that they pertain to the same value, as is the case when blocks hold parts of a value or code blocks rather than replicas.
In this case, the shared storage cost is high independently of the local memory size.
Specifically, we show a bound of \resCommonWriteInvisible{} blocks with invisible readers.
In systems that use symmetric coding (\ie where all blocks are of the same size, namely at least $D/\tau$ bits), this implies a lower bound of $D\cdot 2^D$ bits.
For a modest value size of 20 bytes, the bound amounts to $2.66 \cdot 10^{37} \ \mathrm{TB}$, and for 1KB values it is a whopping $1.02 \cdot 10^{2457} \ \mathrm{TB}$.

We further consider \emph{visible reads}, which can modify the objects' \meta{}.
Such readers may indicate to the writers that a read is ongoing, and signal to them which blocks to retain.
Using such signals, the exponential bound no longer holds -- there are emulations that store a constant number of values per reader~\cite{abraham2007wait, chen2016opodis, peterson1983concurrent, androulaki2014erasure}.
We show that such linear growth with the number of readers is inherent.
Our results are summarized in Table~\ref{tab:results_summary}.

\begin{table}[t]
	\centering
	\begin{tabularx}{\textwidth}{@{}|c|Y|Y|@{}}
		\hline
		& \tbf{Invisible Reads}           & \tbf{Visible Reads}           \\
		\hline
		\twor{\tbf{General Case}} & \twor{\resCommonValueInvisible} & \twor{\resCommonValueVisible} \\
		&                                 &                               \\
		\hline
		\tbf{Common Write}        & \twor{\resCommonWriteInvisible} & \twor{\resCommonWriteVisible} \\
		(\eg coded storage)       &                                 &                               \\
		\hline
	\end{tabularx}
	\caption{Lower bounds on shared storage space consumption, in units of blocks; $D$ is the value size, $\tau > 1$ is the number of data blocks required in order to recover a value, $L \ge 1$ is the maximal number of blocks stored in a reader's local data, and $R$ the number of readers.}
	\label{tab:results_summary}
\end{table}

These bounds are tight as far as regularity and wait-freedom go: relaxing either requirement allows circumventing our results~\cite{faultyMemoryJayanti1998, faultyMemorybyzantineDiscPaxos}.
As for storage cost, our lower bounds are asymptotically tight to known algorithms, whether reads are visible~\cite{peterson1983concurrent, abraham2007wait, androulaki2014erasure} or not~\cite{bazzi2004non, codingGoodson2004efficient, martin2002minimal, dobre2013powerstore}.

We note that the study of the inherent storage blowup in asynchronous coded systems has only recently begun~\cite{spaceBounds, cadambe2016podc} and is still in its infancy.
In this paper, we point out a somewhat surprising similarity between coded storage and other types of shared memory/storage, and show unified lower bounds for all of them.
Section~\ref{sec:discussion} concludes the paper and suggests directions for future work.

\subsection{Related work and applicability of our bounds} \label{subsect:related_work}

Several works have studied the space complexity of register emulations.
Two recent works~\cite{cadambe2016podc, spaceBounds} show a dependence between storage cost and the number of \emph{writers} in crash-tolerant storage, identifying a trade-off between the cost of replication ($f+1$ copies for tolerating $f$ faults) and that of $\tau$-out-of-$n$ coding (linear in the number of writers).
Though they do not explicitly consider disintegrated storage, it is fairly straightforward to adapt the proof from~\cite{spaceBounds} to derive a lower bound of $\tau W$ blocks with $W$ writers.
Here we consider the case of single-writer algorithms, where this bound is trivial.
Other papers~\cite{aguileraDisks2003podc, chocklerSpiegelmanPODC2017} show limitations of multi-writer emulations when objects do not support atomic read-modify-write, whereas we consider single-writer emulations that do use read-modify-write.

Chockler et al.~\cite{chockler2007amnesic} define the notion of \emph{amnesia} for register emulations with an infinite value domain, which intuitively captures the fact that an algorithm ``forgets'' all but a finite number of values written to it.
They show that a wait-free regular emulation tolerating non-authenticated Byzantine faults with invisible readers cannot be amnesic, but do not show concrete space lower bounds.
In this paper we consider a family of disintegrated storage algorithms, with visible and invisible readers, and show concrete bounds for the different cases; if the size of the value domain is unbounded, then our invisible reader bounds imply unbounded shared storage.

Disintegrated storage may also correspond to emulations of large registers from smaller ones, where $\tau$ is the size of the big register divided by the size of the smaller one.
Some algorithms in this vein, e.g.,~\cite{peterson1983concurrent}, indeed have the disintegration property, as the writer writes $\tau$ blocks to a buffer and a reader obtains $\tau$ blocks of the same write.
These algorithms are naturally subject to our bounds.
Other algorithms, e.g.,~\cite{chaudhuri2000one, chen2016opodis, lamportRegular}, do not satisfy our assumption that each block in the shared storage pertains to a single $\Write$ operation, and a reader may return a value based on blocks written by \emph{different} $\Write$ operations.
Thus, our bounds do not apply to them.
It is worth noting that these algorithms nevertheless either have readers signal to the writers and use space linear in the number of readers, or have invisible readers but use space exponential in the value size.
Following an earlier publication of our work, Wei~\cite{wei2018space} showed that these costs -- either linear in the number of visible readers or exponential in the value size with invisible ones -- are also inherent in emulations of large registers from smaller ones that \emph{do} share blocks among writes, albeit do not use meta-data at all.
Several questions remain open in this context: first, Wei's bound is not applicable to all types of storage we consider (in particular, Byzantine), and does not apply to algorithms that use timestamps.
Second, we are not familiar with any regular register emulations where readers write-back data, and it is unclear whether our bound may be circumvented this way.

Non-authenticated Byzantine storage algorithms that tolerate $f$ faults need to read a value $f+1$ times in order to return it, and are thus $\tau$-disintegrated for $\tau = f+1$.
Note that while our model assumes objects are responsive, it a fortiori applies to scenarios where objects may be unresponsive.
Some algorithms circumvent our bound either by providing only safe semantics~\cite{faultyMemoryJayanti1998}, or by forgoing wait-freedom~\cite{faultyMemorybyzantineDiscPaxos}.
Others use channels with unbounded capacity to push data to clients~\cite{martin2002minimal, bazzi2004non} or potentially unbounded storage with best-effort garbage collection~\cite{codingGoodson2004efficient}.

As for coded storage, whenever $\tau$ blocks are required to reconstruct a value, the algorithm is $\tau$-disintegrated.
And indeed, previous solutions in our model require unbounded storage or channels~\cite{codingRashid, codingCachin2006optimal, codingGoodson2004efficient, codingLynchCadambe2014coded, dobre2013powerstore}, or retain blocks for concurrent visible readers, consuming space linear in the number of readers~\cite{androulaki2014erasure}.
Our bounds justify these costs.
Our assumption that each block in the shared storage pertains to a single value is satisfied by almost all coded storage algorithms we are aware of, the only exception is~\cite{wangCadambeMultiversion2014}, which indeed circumvents our lower bound but does not conform to regular register semantics.
Other coded storage solutions, e.g.,~\cite{codingAguilera2005using}, are not subject to our bound because they may recover a value from a single block.	
	
\section{Preliminaries} \label{sec:model}

\myparagraph{Shared storage model}
We consider an asynchronous shared memory system consisting of two types of entities: A finite set $O = \{\bo{1}, \ldots, \bo{n}\}$ of \objects{} comprising \emph{\sharedStorage}, and a set $\Pi$ of processes.
Every entity in the system stores \emph{data}: an \object{}'s data is a single block from some domain $\C$, whereas a process' data is an array of up to $L$ blocks from $\C$.
We assume a bound $L$ on the number of blocks in the data array of each process.
In addition, each entity stores potentially infinite \meta{}, \emph{meta}.
We denote an entity $e$'s data as $e.\data$ and likewise for $e.meta$.
A system's \emph{storage cost} is the number of objects in the \sharedStorage{}, $n$.

\Objects{} support atomic $\get$ and $\update$ \emph{actions} by processes.
We denote by $\op_p$ an action $\op$ performed by $p$ and by $\bobj.\op_p$ an $\op_p$ action at $\bobj$.
An $\bobj.\update_p$ is an arbitrary read-modify-write that possibly writes a block from $\C$ to $\bobj.\data$ and modifies $\bobj.meta$, $p.meta$, and $p.\data$.
An $\bobj.\get_p$ may replace a block in $p.\data$ with $\bobj.\data$ and may modify $p.meta$.

\myparagraph{Algorithms, configurations, and runs}
An \emph{algorithm} defines the behaviors of processes as deterministic state machines, where state transitions are associated with actions.
A \emph{configuration} is a mapping to states (data and meta) from all system components, \ie processes and \objects{}.
In an \emph{initial configuration} all components are in their initial states.

We study algorithms (executed by processes in $\Pi$) that emulate a high-level functionality, exposing high-level operations, and performing low-level $\gets / \updates$ on \objects{}.
We say that high-level operations are \emph{invoked} and \emph{return} or \emph{respond}.
Note that, for simplicity, we model $\gets$ and $\updates$ as instantaneous actions, because the \objects{} are assumed to be atomic, and we do not explicitly deal with \object{} failures in this paper.

A \emph{run} of algorithm $\alg{}$ is a (finite or infinite) alternating sequence of configurations and actions, beginning with some initial configuration, such that configuration transitions occur according to $\alg{}$.
Occurrences of actions in a run are called \emph{events}.
The possible events are high-level operation invocations and responses and $\get$/$\update$ occurrences.
We use the notion of time $t$ during a run $r$ to refer to the configuration reached after the $t^{th}$ event in $r$.
For a finite run $r$ consisting of $t$ events we define $\finalof{r} \triangleq t$.
Two operations are \emph{concurrent} in a run $r$ if both are invoked in $r$ before either returns.
If a process $p$'s state transition from state $\mathcal{S}$ is associated with a low-level action $\op_p \in \{\get_p,\update_p\}$, we say that $\op_p$ is \emph{enabled} in $\mathcal{S}$.
A run $r'$ is an \emph{extension} of a (finite) run $r$ if $r$ is a prefix of $r'$;
we denote by \suffix{r'}{r} the suffix of $r'$ that starts at $\finalof{r}$.
If a high-level operation $op$ has been invoked by process $p$ but has not returned by time $t$ in a run $r$, we say that $op$'s invocation is \emph{pending} at $t$ in $r$. We assume that each process' first action in a run is an invocation, and a process has at most one pending invocation at any time.

For $e \in \Pi \cup O$, we denote by $e.\data(r,t)$ the set of distinct blocks stored in $e.\data$ at time $t$ in a run $r$.
Since for an object $o$, $\abs{\odata{r,t}} = 1$, we sometimes refer to $\odata{r,t}$ as the block itself, by slight abuse of notation.
We say that $p$ \emph{obtains} a block $b$ at time $t$ in a run $r$, if $b \notin \pdata{r,t}$ and $b \in \pdata{r,t+1}$.

\myparagraph{Register emulations}
We study algorithms that emulate a shared \emph{register}~\cite{lamportRegular}, which stores a value $v$ from some domain $\V$.
We assume that $\abs{\V} = 2^D > 1$, \ie values can be represented using $D > 0$ bits.
For simplicity, we assume that each run begins with a dummy initialization operation that writes the register's initial value and does not overlap any operation.
The register exposes high-level $\Read_p$ and $\Write_p(v)$ operations of values $v \in \V$ to processes $p \in \Pi$.
We consider single-writer (SW) registers where the application at only one process (the \emph{writer}) invokes $\Writes$, and hence omit the subscript $p$ from $\wr{v}$.
The remaining $R \triangleq \abs{\Pi} - 1$ processes are limited to performing $\Reads$, and are referred to as \emph{readers}.
For brevity, we refer to the subsequence of a run where a specific invocation of a $\wr{v} / \Read_p$ is pending simply as a $\wr{v} / \Read_p$ operation.

We assume that whenever a $\Read_p$ operation is invoked at time $t$ in a run $r$, $\pdata{r,t}$ is empty.
We consider two scenarios: (1) \emph{invisible reads}, where $\Reads$ do not use $\updates$, and (2) \emph{visible reads}, where $\Reads$ may perform $\updates$ that update \meta{} (only) in the shared storage.
Note that readers do \emph{not} write actual data, which is usually the case in regular register emulations, defined below.
In a single-reader (SR) register $R=1$, and if $R>1$ the register is multi-reader (MR).
If the states of the writer and the \objects{} at the end of a finite run $r$ are equal to their respective states at the end of a finite run $r'$, we say that $\finalof{r}$ and $\finalof{r'}$ are \emph{indistinguishable} to the writer and \objects{}, and denote: $\finalof{r} \indist \finalof{r'}$.

Our safety requirement is \emph{regularity}~\cite{lamportRegular}: a $\Read$ $rd$ must return the value of either the last $\Write$ $w$ that returns before $rd$ is invoked, or some write that is concurrent with $rd$.
For liveness, we require \emph{wait-freedom}, namely that every operation invoked by a process $p$ returns within a finite number of $p$'s actions.
In other words, if $p$ is given infinitely many opportunities to perform actions, it completes its operation regardless of the actions of other processes.
	
\section{Disintegrated storage} \label{sec:disintegrated_storage}

As noted above, existing wait-free algorithms of coded and/or Byzantine-fault-tolerant storage with invisible readers may store all values ever written~\cite{bazzi2004non, codingGoodson2004efficient, martin2002minimal, codingCachin2006optimal, cadambe2016podc, codingRashid, dobre2013powerstore}.
This is because if old values are erased, it is possible for a slow reader to never find sufficiently many blocks of the same value so as to be able to return it.
If readers are visible, then a value per reader is retained.
We want to prove that these costs are inherent.
The challenge in proving such space lower bounds is that the aforementioned algorithms use unbounded timestamps.
How can we show a space lower bound if we want to allow algorithms to use unbounded timestamps?
We address this by allowing \meta{} to store timestamps, etc., and by not counting the storage cost for \meta{}.
For example, the above algorithms store timestamps in \meta{} alongside data blocks and use them to figure out which data is safe to return, but still need $\tau$ actual blocks/copies of a value in order to return it.
Note that for the sake of the lower bound, we do not restrict how \meta{} is used; all we require is that the algorithm read $\tau$ data blocks of the same value (or $\Write$), and we do not specify how the algorithm knows that they pertain to the same value (or $\Write$).
To formalize the property that the algorithm returns $\tau$ blocks pertaining to the same value or $\Write$, we need to track, for each block in the shared storage, which $\Write$ produced it.
To this end, we define \emph{labels}.
Labels are only an analysis tool, and do not exist anywhere.
In particular, they are not timestamps, not \meta{}, and not explicitly known to the algorithm.
As an external observer, we may add them as abstract state to the blocks, and track how they change.

\myparagraph{Labels}
We associate each block $b$ in the shared or local storage with a set of labels, $\sharedlabel{b}$, as we now explain.
For an algorithm $\alg{}$ and $v \in \V$, denote by $\W_v^{\alg}$ the set of $\wr{v}$ operations invoked in runs of $\alg{}$.
For $V \subseteq \V$, we denote $\W_V^{\alg} \triangleq \bigcup_{v \in V} \W_v^{\alg}$, and let $\W^{\alg} \triangleq \W_{\V}$.
For clarity, we omit $\alg{}$ when obvious from the context, and refer simply to $\W_v$, $\W_V$, and $\W$.
We assume that the $k^{th}$ $\update$ event occurring in a $\Write$ operation $w \in \W$ tags the block $b$ it stores (if any) with a unique label $\mylabel{w,k}$, so $\Label{b}$ becomes $\{ \mylabel{w,k}\}$.

Whereas our assumption that each block in the shared storage pertains to a single write rules out associating multiple labels with such a block, we do allow the reader's \meta\ to recall multiple accesses encountering the same block.
For example, when blocks are copies of a replicated value, the reader can store one instance of the value in local memory and keep a list of the objects where the value was encountered.
To this end, a block in a reader's \textit{data} may be tagged with multiple labels:
when a reader $p$ obtains a block $b$ from an object $o$ at time $t$ in a run $r$, the block $b$ in $\pdata{r,t+1}$ is tagged with $\Label{\odata{r,t}}$;
if at time $t' > t$ $p.\data$ still contains $b$ and $p$ performs an action on an object $o'$ s.t.\ $\pdata[o']{r,t'} = b$ and the latter is tagged with label $\ell$, $p$ adds $\ell$ to $\Label{b}$ (regardless of whether $b$ is added to $p.\data$ once more).
When all copies of a block are removed from $p.\data$, all its labels are ``forgotten''.
We emphasize that labels are not stored anywhere, and are only used for analysis.

We track the labels of a value $v \in \V$ at time $t$ in a run $r$ using the sets $\slabels{v,r,t}$, of labels in the shared storage, $\llabels{v,r,t}$, of labels in process $p$'s local storage, and $\alllabels{v,r,t}$, a combination of both.
Formally,
\begin{itemize}
	\item
	$\slabels{v,r,t} \triangleq \left( \bigcup_{o \in O} \Label{\odata{r,t}} \right) \cap \left( \W_v \times \N \right)$.
	\item
	$\llabels{v,r,t} \triangleq \left( \bigcup_{b \in \pdata{r,t}} \Label{b} \right) \cap \left( \W_v \times \N \right)$.
	\item
	$\alllabels{v,r,t} \triangleq \llabels{v,r,t} \cup \slabels{v,r,t}$.
\end{itemize}
For a time $t$ in a run $r$ and $p \in \Pi$, we define
$\values{r,t} \triangleq \{ v \in \V \; \mid \; \llabels{v,r,t} \ne \emptyset \}$.

Similarly, we track labels associated with a particular $\Write$ $w \in \W$
accessible by process $p \in \Pi$ at time $t$ in a run $r$:
\begin{itemize}
	\item
	$\slabels{w,r,t} \triangleq \left( \bigcup_{o \in O} \Label{\odata{r,t}} \right) \cap \left( \{w\} \times \N \right)$.
	\item
	$\llabels{w,r,t} \triangleq \left( \bigcup_{b \in \pdata{r,t}} \Label{b} \right) \cap \left( \{w\} \times \N \right)$.
	\item
	$\alllabels{w,r,t} \triangleq \llabels{w,r,t} \cup \slabels{w,r,t}$.
\end{itemize}
We define $\writes{r,t} \triangleq \{ w \in \W \; \mid \; \llabels{w,r,t} \ne \emptyset \}$.
Note that for all $v \in \V$ and $w \in \W_v$, (1) $\slabels{w,r,t} \subseteq \slabels{v,r,t}$, (2) $\llabels{w,r,t} \subseteq \llabels{v,r,t}$, and (3) $\alllabels{w,r,t} \subseteq \alllabels{v,r,t}$.

Since readers do not write-back:
\begin{observation} \label{obs:no_new_labels}
	If the $t^{th}$ event in a run $r$ is of a reader $p \in \Pi$, then for all $v \in \V, w \in \W$: $\alllabels{v,r,t} \subseteq \alllabels{v,r,t-1}$ and $\alllabels{w,r,t} \subseteq \alllabels{w,r,t-1}$.
\end{observation}

\myparagraph{Disintegrated storage}
Intuitively, in disintegrated storage register emulations, for a $\Read_p$ to return $v$, $p$ must encounter $\tau > 1$ blocks corresponding to $v$ that were produced by separate $\update$ events.
To formalize this, we use labels:
\begin{definition}[\disint{}]
	If a return of $v \in \V$ by a $\Read_p$ invocation is enabled at time $t$ in a run $r$ then $\abs{\llabels{v,r,t}} \ge \tau$.
\end{definition}
Thus, a reader can only return $v$ if it recalls (in its local memory) obtaining blocks of $v$ with $\tau$ different labels.

A more restrictive case of \disint{} occurs when readers cannot identify whether two blocks pertain to a common value unless they are produced by a common write that identifies them, \eg with the same timestamp.
This is the case when value parts or code words are stored in \objects{} rather than full replicas.

To capture this case, for a block $b \in \bigcup_{e \in O \cup \Pi} e.\data$, a value $v \in \V$, and a $\Write$ $w \in \W_v$, if $\exists k \in \N$ s.t.\ $\mylabel{w,k} \in \Label{b}$, we say that $w$ is an \emph{origin write} of $b$ and $v$ is an \emph{origin value} of $b$.
Common write $\tau$-disintegrated storage is then defined:
\begin{definition}[common write $\tau$-disintegrated storage]
	If a return of $v \in \V$ by a $\Read_p$ invocation is enabled at time $t$ in a run $r$ then $\exists w \in \W_v \; : \; \abs{\llabels{w,r,t}} \ge \tau$.
\end{definition}
Note that we do not further require $p.\data$ to actually hold $\tau$ blocks with a common write, because the weaker definition suffices for our lower bounds.
For brevity, we henceforth refer to a common write \disint{} algorithm simply as \commonWrite{}.

\myparagraph{Permanence}
Our lower bounds will all stem, in one way or another, from the observation that in wait-free disintegrated storage, every run must reach a point after which some values (and in the case of common write, also some $\Writes$) must \emph{permanently} have a certain number of blocks in the shared storage.
This is captured by the following definition:
\begin{definition}[permanence]
	Consider a finite run $r$, $k \in \N$, a set $S \subseteq \V$, and a set of readers $\Theta \subset \Pi$.
	Let $z \in \V \cup \W$ be a value or a $\Write$ operation.
	We say that $z$ is \emph{$\permanent{k}{S}{\Theta}$} in $r$ if in every finite extension $r'$ of $r$ s.t.\ in \suffix{r'}{r} readers in $\Theta$ do not take actions and $\Writes$ are limited to values from $S$, $\abs{\slabels{z,r',\finalof{r'}}} \ge k$.
\end{definition}
Intuitively, this means that the shared storage continues to hold $k$ blocks of $z$ as long as readers in $\Theta$ do not signal to the writer and only values from $S$ are written.
For brevity, when the particular sets $S$ and $\Theta$ are not important, we refer to the value shortly as $k$-permanent.
The observation below follows immediately from the definition of permanence:
\begin{observation} \label{obs:write-permanence-implies-value-permanence} \label{obs:permanence_in_constructions}
	Let $v \in \V$, $w \in \W_v$, $k \in \N$, $V_2 \subseteq V_1 \subseteq \V$, $\Theta_1 \subseteq \Theta_2 \subset \Pi$.
	\begin{enumerate}
		\item
		If $w$ is $\permanent{k}{V_1}{\Theta_1}$ in a finite run $r$ then $v$ is $\permanent{k}{V_1}{\Theta_1}$ in $r$.
		\item 
		If $v$ is $\permanent{k}{V_1}{\Theta_1}$ in a finite run $r$ then $v$ is $\permanent{k}{V_2}{\Theta_2}$ in all finite extensions $r'$ of $r$ where in \suffix{r'}{r} $\Writes$ are limited to values from $V_1$ and readers in $\Theta_1$ do not take actions.
	\end{enumerate}
\end{observation}
Since each object holds a single block associated with a single label:
\begin{observation} \label{obs:n_ge_slabels}
	For time $t$ in a run $r$, the number of objects is:
	$ n \ge \abs{\bigcup_{v \in \V} \slabels{v,r,t}}. $
\end{observation}
Thus, if there are $m$ different $k$-permanent values in a run, then $n \ge mk$.
We observe that with invisible readers, the set $\Theta$ is immaterial:
\begin{claim} \label{claim:permanence_with_invisible_reader}
	Consider $k \in \N$, $V \subseteq \V$, and a finite run $r$ with an invisible reader $p \in \Pi$.
	If $z \in \V \cup \W$ is $\permanent{k}{V}{\{p\}}$ in $r$ then $z$ is $\permanent{k}{V}{\emptyset}$ in $r$.
\end{claim}
\begin{proof}
	Assume by contradiction that there exists an extension $r'$ of $r$ where in \suffix{r'}{r} $\Writes$ are limited to values from $V$, $p$ takes steps, and $\abs{\slabels{z,\rfinal{r'}}} < k$.
	Let $r''$ be the extension of $r$ identical to $r'$ except in that $p$ does not take steps in \suffix{r''}{r}.
	Since $p$ is invisible, $\finalof{r''} \indist{\finalof{r'}}$, thus $\abs{\slabels{z,\rfinal{r''}}} < k$, in contradiction to $v$ being $\permanent{k}{V}{\{p\}}$ in $r$.
\end{proof}

The specific lower bounds for the four scenarios we consider differ in the number of permanent values/$\Writes$ and the number of blocks per value/$\Write$ ($k = \tau-1$ or $k = \tau$) we can force the shared storage to retain forever in each case.
Interestingly, our notion of permanence resembles the idea that an algorithm is not \emph{amnesic} introduced in~\cite{chockler2007amnesic} (see Section~\ref{subsect:related_work}), but is more fine-grained in specifying the number of permanent blocks and restricting executions under which they are retained.

\section{Lower bounds for disintegrated storage}
\label{sec:commonValue}

In this section we provide lower bounds on the number of objects required for \disint{} regular wait-free register emulations.
Section~\ref{subsect:commonValueGeneralObservations} proves two general properties of regular wait-free \disint{} algorithms.
We show in Section~\ref{subsect:commonValueInvisibleReads} that with invisible $\Reads$, unless the readers' local storage size is exponential in $D$, the storage cost of such emulations is at least exponential in $D$.
Finally, Section~\ref{subsect:commonValueVisibleReads} shows that if $\Reads$ are visible, then the storage cost increases linearly with the number of readers.

\subsection{General properties} \label{subsect:commonValueGeneralObservations}

We first show that because readers must make progress even if the writer stops taking steps, at least $2\tau - 1$ blocks are required regardless of the number of readers.

\begin{claim} \label{claim:2tauMinusOne}
	Consider $v_1,v_2 \in \V$ and a run $r$ of a wait-free regular \disint{} algorithm with two consecutive responded $\Writes$ $w_1 \in \W_{v_1}$ followed by $w_2 \in \W_{v_2}$.
	Let $p \in \Pi$ be a reader s.t.\ no $\Read_p$ is pending in $r$.
	Then there is a time $t$ between the returns of $w_1$ and $w_2$ when $\abs{\slabels{v_1,r,t}} \ge \tau$ and $\abs{\slabels{v_2,r,t}} \ge \tau - 1$.
\end{claim}

\begin{proof}
	We first argue that at the time $t_i$, $i \in \{1,2\}$ when $w_i$ returns, $\abs{\slabels{v_i,r,t_i}} \ge \tau$.
	Assume the contrary.
	We build a run $r'$ identical to $r$ up to $t_i$.
	In $r'$, only process $p$ performs actions after time $t_i$.
	Next, invoke a $\Read_p$ operation $rd$.
	By regularity and wait-freedom, $rd$ must return $v_i$.
	Before performing actions on objects, $\pdata{r',t_i}$ is empty, thus, from \disint{}, $p$ must encounter at least $\tau$ blocks with an origin value of $v_i$ in order to return it.
	Since no process other than $p$ takes actions, $\abs{\slabels{v_i,r',t'}} < \tau$ for all $t' \ge t_i$ onward, so $rd$ cannot find these blocks and does not return $v_i$, a contradiction.
	It follows that in $r'$ at $t_i$, and hence also in $r$ at $t_i$, $\abs{\slabels{v_i,r,t_i}} \ge \tau$.
	
	Next, if at $t_1$, $\abs{\slabels{v_2,r,t_1}} \ge \tau-1$ then we are done.
	Otherwise, observe that \objects{} are accessed one-at-a-time.
	Therefore, and since $\abs{\slabels{v_2,r,t_1}} < \tau-1$, there exists a time $t$ between $t_1$ and $t_2$ when $\abs{\slabels{v_2,r,t}} = \tau - 1$.
	
	Finally, assume that $\abs{\slabels{v_1,r,t}} < \tau$.
	Build a run $r''$ identical to $r$ up to $t$, where again only $p$ takes actions after $t$.
	As above, it follows by regularity, \disint{}, and $\pdata{r'',t} = \emptyset$, that $rd$ never returns, in violation of wait-freedom.
	It follows that $\abs{\slabels{v_1,r'',t}} = \abs{\slabels{v_1,r,t}} \ge \tau$.
	\qedhere
\end{proof}

The following lemma states that every non-empty set $V$ can be split into two disjoint subsets, where one contains a value that is $(\tau-1)$-permanent with respect to the other subset.
The idea is to show that in the absence of such a value, a reader's accesses to the shared storage may be scheduled in a way that prevents the reader from obtaining $\tau$ labels of the same value.
The logic of the proof is the following:
we restrict $\Writes$ to a set of values $V$, and consider the set $S$ of values with blocks in $p.\data \cap V$.
If no value in $S$ is $(\tau-1)$-permanent, then we can bring the shared storage to a state where none of the values in $S$ have $\tau$ labels, preventing the reader from obtaining the $\tau$ labels required to return.
By regularity, readers cannot return other values.
The formal proof is slightly more subtle, because it needs to consider $\llabelsnoarg_p$ as well as labels in the shared storage.
It shows that the total number of labels of values in $S$ (in both the shared and local storage) remains below $\tau$ whenever $p$ takes a step.

\begin{lemma} \label{lem:burning_lemma}
	Consider a non-empty set of values $V \subseteq \V$, a set of readers $\Theta \subset \Pi$, a reader $p \in \Pi \setminus \Theta$, and a finite run $r$ of a wait-free regular \disint{} algorithm.
	Then there is a subset $S \subseteq V$ of size $1 \le \abs{S} \le L$ and an extension $r'$ of $r$ where some value $v \in S$ is $\permanent{\tau-1}{V \setminus S}{\Theta \cup \{p\}}$ and s.t.\ in \suffix{r'}{r} $\Writes$ are limited to values from $V$ and readers in $\Theta$ do not take steps.
\end{lemma}

\begin{proof}
	Assume by contradiction that the lemma does not hold.
	We construct an extension $r'$ of $r$ where a $\Read_p$ operation includes infinitely many actions of $p$ yet does not return.
	To this end, we show that the following property holds at specific times in \suffix{r'}{r}:
	\[ \invariantValues{\hat{r},t} \triangleq \forall v \in \values{\hat{r},t} \cap V \, : \, \abs{\alllabels{v,\hat{r},t} } < \tau. \]
	
	First, extend $r$ to $r_0$ by returning any pending $\Read_p$ and $\Write$, invoking and returning a $\wr{v_0}$ for some $v_0 \in \V$ (the operations eventually return, by wait-freedom), and finally invoking a $\Read_p$ operation $rd$ without allowing it to take actions.
	We now prove by induction that for all $k \in \N$, there exists an extension $r'$ of $r_0$ where (1) $\invariantValues{\rfinal{r'}}$ holds and in \suffix{r'}{r}: (2) $\Writes$ are restricted to values from $V$, (3) $p$ performs $k$ actions on objects following $rd$'s invocation, and (4) $rd$'s return is not enabled, and (5) processes in $\Theta$ do not take steps.
	
	Base: for $k = 0$, consider $r' = r_0$.
	(3,5) hold trivially.
	(2) holds since the only $\Write$ in \suffix{r'}{r} is of $v_0 \in V$.
	Since $p$ performs no actions following the invocation of $rd$, $\pdata{r',\finalof{r'}}$ is empty.
	Therefore, (1) $\invariantValues{\rfinal{r'}}$ is vacuously true, and $\llabels{v,r,t}$ is empty for all $v \in \V$, thus (4) $rd$'s return is not enabled by \disint{}.
	
	Step: assume inductively such an extension $r_1$ of $r_0$ with $k \ge 0$ actions performed by $p$ following $rd$'s invocation.
	Since $rd$ cannot return, by wait-freedom, an action $\step$ is enabled on some \object{}.
	We construct an extension $r_2$ of $r_1$ by letting $\step$ occur at time $\finalof{r_1}$.
	We consider two cases:
	
	1. $p$ does not obtain a block with an origin value in $V \setminus \values{r_1, \finalof{r_1}}$ at $\step$, thus $\values{r_2,\finalof{r_2}} \cap V \subseteq \values{r_1,\finalof{r_1}} \cap V$.
	Then, by Observation~\ref{obs:no_new_labels} and the inductive hypothesis, (1) $\invariantValues{r_2,\finalof{r_2}}$ holds and thus, by \disint{}, $rd$ cannot return any value $v \in \values{r_2,\finalof{r_2}} \cap V$ at $\finalof{r_2}$.
	(4) It cannot return any other value in $\values{r_2,\finalof{r_2}}$ by regularity, and $r_2$ satisfies the induction hypothesis for $k+1$, as (2,3,5) trivially hold.
	
	2. $p$ obtains a block with origin value $u \in V \setminus \values{r_1,\finalof{r_1}}$ at time $\finalof{r_1}$.
	Then $\abs{\llabels{u,r_2,\finalof{r_2}}} = 1$.
	By Observation~\ref{obs:no_new_labels} and the inductive hypothesis, for all $v \in \values{\rfinal{r_2}} \setminus \{u\}$, $\abs{\llabels{v,\rfinal{r_2}}}<\tau$, and thus $rd$'s return is not enabled at time $\finalof{r_2}$ by \disint{} and regularity.
	
	Let $S = \values{\rfinal{r_2}} \cap V$, and note that $\abs{S} \ge 1$ (since $u \in S$) and that $\abs{S} \le \abs{p.\data} \le L$.
	By the contradicting assumption, $u$ is not $\permanent{\tau-1}{V \setminus S}{\Theta \cup \{p\}}$ in $r_2$, thus there exists an extension $r_3$ of $r_2$ s.t.\ $\abs{\slabels{u,\rfinal{r_3}}} < \tau-1$ and in \suffix{r_3}{r_2} $\Writes$ are limited to values from $V \setminus S$ and no readers in $\Theta \cup \{p\}$ take steps (3,5 hold).
	Since $p$ takes no steps in \suffix{r_3}{r_2}, we have that $\llabels{u,\rfinal{r_3}} = \llabels{u,\rfinal{r_2}}$, yielding:
	\begin{equation} \label{eq:Phi_holds_for_u}
	\abs{\alllabels{u,\rfinal{r_3}}} \le \abs{\llabels{u,\rfinal{r_2}}} + \abs{\slabels{u,\rfinal{r_3}}} < 1 + (\tau-1) = \tau.
	\end{equation}
	
	All $\Writes$ invoked after $\finalof{r_2}$ are from $\W_{V \setminus S}$ (2 holds), and therefore do not produce new labels associated with values in $S$.
	Since no values in $S$ are written after $\finalof{r_1}$ and readers' actions do not affect the sets $\slabelsnoarg$, by Observation~\ref{obs:no_new_labels}, we have that $\forall v \in S$, $\alllabels{v,\rfinal{r_3}} \subseteq \alllabels{v,r_1,\finalof{r_1}}$, and since $\invariantValues{r_1,\finalof{r_1}}$ holds (inductively) and $S \setminus \{u\} \subseteq \values{r_1,\finalof{r_1}} \cap V$,
	\begin{equation} \label{eq:Phi_holds_for_S_minus_u}
	\forall v \in S \setminus \{u\} \, : \, \abs{\alllabels{v,\rfinal{r_3}}} \le \abs{\alllabels{v,r_1,\finalof{r_1}}} < \tau.
	\end{equation}
	From Equations~\ref{eq:Phi_holds_for_u} and~\ref{eq:Phi_holds_for_S_minus_u}, and since $\values{r_3,\finalof{r_3}} \cap V = \values{r_2,\finalof{r_2}} \cap V = S$, we get $\invariantValues{r_3,\finalof{r_3}}$ (1).
	Since $rd'$s return was not enabled at time $\finalof{r_2}$ and it took no actions since, its return is still not enabled (4), and we are done.
	\qedhere
\end{proof}

\subsection{Invisible reads} \label{subsect:commonValueInvisibleReads}
We now consider a setting of a single reader and single writer where $\Reads$ are invisible.
To show the following theorem, we ``blow up'' the shared storage by repeatedly invoking Lemma~\ref{lem:burning_lemma}, each time adding one more $(\tau-1)$-permanent value, yielding the following bound:

\begin{theorem} \label{thm:common_value_invisible_reads_bound}
	The storage cost of a regular \disint{} wait-free SRSW register emulation where $\Reads$ are invisible is at least \resCommonValueInvisible\ blocks.
\end{theorem}

When readers are invisible, the set $\Theta$ is of no significance, so we consider $\emptyset$.
Given a set of values $V$, the value added by Lemma~\ref{lem:burning_lemma} is $\permanent{\tau-1}{V \setminus S}{\emptyset}$ for a smaller set of values $V \setminus S$ where $\abs{S} \le L$.
Therefore, we can invoke Lemma~\ref{lem:burning_lemma} $m = \mdef$ times before running out of values, showing the following:

\begin{lemma} \label{lem:construction_common_value_invisible_reads}
	Let $p \in \Pi$ be an invisible reader.
	There exist finite runs $r_0,...,r_{m}$ and sets of values $V_0 \supset V_1 \supset ... \supset V_{m}$ and $U_0 \subset U_1 \subset ... \subset U_m$, such that for all $0 \le k \le m$:
	\begin{enumerate}
		\item $\abs{V_k} \ge 2^D - L k$, $\abs{U_k} = k$, $V_k \cap U_k = \emptyset$, and
		\item all elements of $U_k$ are $\permanent{\tau-1}{V_k}{\emptyset}$ in $r_k$.
	\end{enumerate}
\end{lemma}

\begin{proof}
	By induction.
	Base: $r_0$ is the empty run, $V_0 = \V$ and $U_0 = \emptyset$.
	Assume inductively that the lemma holds for $k < m$.
	Since $m < \mneedceil$, we get: $ |V_k| > 2^D - L \mneedceil = 1. $
	Since $V_k$ is non-empty and $\abs{\emptyset} < R$, by Lemma~\ref{lem:burning_lemma} there exist an extension $r_{k+1}$ of $r_k$ where $\Writes$ in \suffix{r_{k+1}}{r_k} are limited to values from $V_k$, a set $S \subset V_k$, $1 \le \abs{S} \le L$, and a value $v \in S$ that is $\permanent{\tau-1}{V_k \setminus S}{\{p\}}$ in $r_{k+1}$.
	
	Let $V_{k+1} = V_k \setminus S$ and $U_{k+1} = U_{k} \cup \{v\}$.
	Note that, because $V_k \cap U_k = \emptyset$ and $v \in S \subset V_k$, we get that $V_{k+1} \cap U_{k+1} = \emptyset$ and $\abs{U_{k+1}} = \abs{U_k} + 1 = k+1$.
	Since $1 \le \abs{S} \le L$ we have that $V_k \supset V_{k+1}$ and $\abs{V_{k+1}} \ge \abs{V_k} - \abs{S} \ge 2^D - L (k+1)$.
	By the inductive assumption and Observation~\ref{obs:permanence_in_constructions}, all values in $U_k$ are $\permanent{\tau-1}{V_{k+1}}{\emptyset}$ in $r_{k+1}$.
	By Claim~\ref{claim:permanence_with_invisible_reader}, $v$ is also $\permanent{\tau-1}{V_{k+1}}{\emptyset}$ in $r_{k+1}$ and we are done.
	\qedhere
\end{proof}

Our bound combines the $2\tau-1$ blocks of Claim~\ref{claim:2tauMinusOne} with the $(\tau-1)m$ from Lemma~\ref{lem:construction_common_value_invisible_reads}:

\begin{proof}[Proof (Theorem~\ref{thm:common_value_invisible_reads_bound})]
	Consider an invisible reader $p \in \Pi$ and construct $r_m$, $V_m$, and $U_m$ as in Lemma~\ref{lem:construction_common_value_invisible_reads}.
	Note that $V_m$ contains at least two distinct values that are not in $U_m$, since $V_m \cap U_m = \emptyset$ and $\abs{V_m} \ge 2^D - Lm > 2^D - L \mneedceil = 1$.
	Extend $r_m$ to $r_{m+1}$ by invoking and returning $\wr{v}$ and $\wr{v'}$ for $v,v' \in V_{m}$.
	
	By Claim~\ref{claim:2tauMinusOne}, there is a time $t \ge \finalof{r_m}$ in $r_{m+1}$ when there are $2\tau-1$ blocks in the \sharedStorage{} with origin values of $v$ or $v'$.
	In addition, by Lemma~\ref{lem:construction_common_value_invisible_reads}, $U_m$ consists of $m$ values that are $\permanent{\tau-1}{V_m}{\emptyset}$ in $r_m$, and since $\Writes$ in \suffix{r_{m+1}}{r_m} are of values from $V_m$, the values in $U_m$ remain $\permanent{\tau-1}{V_m}{\emptyset}$ in $r_{m+1}$.
	By Observation~\ref{obs:n_ge_slabels}:
	\[ n
	\ge 2\tau - 1 + (\tau-1) m
	=        \tau + (\tau-1)(m+1)
	= \resCommonValueInvisible.
	\qedhere
	\]
\end{proof}

\subsection{Visible reads} \label{subsect:commonValueVisibleReads}

We now consider systems where readers may write \meta{} in the shared storage.
We use a similar technique as in Lemma~\ref{lem:construction_common_value_invisible_reads}, except that due to readers' $\updates$, the indistinguishability argument can no longer be used.
Instead, we invoke a new reader for each extension, and therefore the number of runs might be limited by the number of readers, $R$:

\begin{theorem} \label{thm:common_value_visible_reads_bound}
	The storage cost of a regular \disint{} wait-free MRSW register emulation with $R$ readers is at least \resCommonValueVisible{} blocks.
\end{theorem}

To achieve this bound, we use Lemma~\ref{lem:burning_lemma} again to construct $N = \commonValueN$ extensions of the empty run (note that it does not assume invisible $\Reads$).

\begin{lemma} \label{lem:construction_common_value_visible_reads}
	There exist finite runs $r_0,...,r_{N}$, sets of values $V_0 \supset V_1 \supset ... \supset V_{N}$ and $U_0 \subset U_1 \subset ... \subset U_N$, and sets of readers $\Theta_0 \subset \Theta_1 \subset ... \subset \Theta_N$, such that for all $0 \le k \le N$:
	\begin{enumerate}
		\item $\abs{V_k} \ge 2^D - L k$,  $\abs{U_k} = \abs{\Theta_k} = k$,  $V_k \cap U_k = \emptyset$, and
		\item all elements of $U_k$ are $\permanent{\tau-1}{V_k}{\Theta_k}$ in $r_k$.
	\end{enumerate}
\end{lemma}

\begin{proof}
	By induction.
	Base: $r_0$ is the empty run, $V_0 = \V$, $\Theta_0 = U_0 = \emptyset$.
	Assume inductively such $r_k$, $V_k$, $U_k$, and $\Theta_k$ for $k < N$, and construct $r_{k+1}$ as follows:
	since $R - \abs{\Theta_k} > 0$, there is a reader $p \in \Pi \setminus \Theta_k $.
	Since $N < \mneedceil$, we get $|V_k| > 2^D - L N > 1$.
	Therefore, by Lemma~\ref{lem:burning_lemma}, there exist an extension $r_{k+1}$ of $r_k$ where in \suffix{r_{k+1}}{r_k} $\Writes$ are limited to values from $V_k$ and readers in $\Theta_k$ do not take steps, a set $S \subseteq V_k$, $1 \le \abs{S} \le L$, and a value $v \in S$ that is $\permanent{\tau-1}{V_k \setminus S}{\Theta_k \cup \{p\}}$ in $r_{k+1}$.
	
	Let $V_{k+1} = V_k \setminus S$ and $U_{k+1} = U_{k} \cup \{v\}$.
	Note that, because $V_k \cap U_k = \emptyset$ and $v \in S \subset V_k$, it follows that $V_{k+1} \cap U_{k+1} = \emptyset$ and $\abs{U_{k+1}} = k+1$.
	Furthermore, since $1 \le \abs{S} \le L$, we get: $V_k \supset V_{k+1}$ and $\abs{V_{k+1}} \ge \abs{V_k} - \abs{S} \ge 2^D - L (k+1)$.
	Finally, let $\Theta_{k+1} = \Theta_k \cup \{p\}$.
	By the inductive assumption and Observation~\ref{obs:permanence_in_constructions}, all values in $U_k$ are $\permanent{\tau-1}{V_{k+1}}{\Theta_{k+1}}$ in $r_{k+1}$, and so all of $U_{k+1}$ is $\permanent{\tau-1}{V_{k+1}}{\Theta_{k+1}}$ in $r_{k+1}$, as needed.
	\qedhere
\end{proof}

From Lemma~\ref{lem:construction_common_value_visible_reads}, in $r_N$ there is a set of $N$ $(\tau-1)$-permanent values, inducing a cost of $(\tau-1)N$.
We use Claim~\ref{claim:2tauMinusOne} to increase the bound by $2\tau-1$ additional blocks.

\begin{proof}[Proof (Theorem~\ref{thm:common_value_visible_reads_bound})]
	Construct $r_N$, $V_N$, $U_N$, and $\Theta_N$ as in Lemma~\ref{lem:construction_common_value_visible_reads}.
	Note that, since $R - N \ge 1$, there exists $p \in \Pi \Theta_N $.
	Since $V_N \cap U_N = \emptyset$ and $\abs{V_N} \ge 2^D - LN > 2^D - L \mneedceil = 1$, $V_N \setminus U_N$ contains at least two values.
	Extend $r_N$ to $r_{N+1}$ by invoking and returning $\wr{v}$ and $\wr{v'}$ for $v,v' \in V_{N} \setminus U_N$.
	
	By Claim~\ref{claim:2tauMinusOne}, there is a time $t \ge \finalof{r_N}$ in $r_{N+1}$ when there are $2\tau-1$ blocks in the \sharedStorage{} with origin values of $v$ or $v'$.
	$U_N$ consists of $N$ additional values that are $\permanent{\tau-1}{V_N}{\Theta_N}$ in $r_N$, and since in \suffix{r_{N+1}}{r_N} $\Writes$ are of values from $V_N$ and no reader in $\Theta_N$ takes steps, the values in $U_N$ remain $\permanent{\tau-1}{V_N}{\Theta_N}$ in $r_{N+1}$.
	By Observation~\ref{obs:n_ge_slabels}, the storage cost is:
	\[ n \ge 2\tau - 1 + (\tau-1) N = \tau + (\tau-1)(N+1) = \resCommonValueVisible.
	\qedhere
	\]
\end{proof}

\section{Lower bounds for common write disintegrated storage}
\label{sec:commonWrite}

While the results of the previous section hold a fortiori for \commonWrite{} algorithms, in this case we are able to show stronger results, independent of the local storage size.
Intuitively, this is because readers can no longer reuse blocks they obtained from previous $\Writes$ of the same value, and so we can prolong the execution that blows up the shared storage by rewriting values.
Section~\ref{subsect:commonWriteGeneralObservations} proves a general attribute of \commonWrite{} algorithms. 
We show in Section~\ref{subsect:commonWriteInvisibleReads} that even with a single reader (and a single writer), if $\Reads$ are invisible, then the required storage cost is at least \resCommonWriteInvisible.
In Section~\ref{subsect:commonWriteVisibleReads} we prove a bound for visible $\Reads$.

\subsection{General observation} \label{subsect:commonWriteGeneralObservations}

In this section we define a property that is a special case of $k$-permanence, which additionally requires that the set of labels associated with a $\Write$ does not change.
\begin{definition}[Constancy]
	Consider a finite run $r$, $k \in \N$, a set $S \subseteq \V$, and a set of readers $\Theta \subset \Pi$.
	We say that a $\Write$ $w \in \W$ is $\constant[k]{S}{\Theta}$ in $r$ if in every finite extension $r'$ of $r$ s.t.\ in \suffix{r'}{r} readers in $\Theta$ do not take actions and $\Writes$ are limited to values from $S$, $\slabels{w,r', \finalof{r'}} = \slabels{w, r, \finalof{r}}$ and $\abs{\slabels{w,r',\finalof{r'}}} = k$.
\end{definition}
Similarly to Claim~\ref{claim:permanence_with_invisible_reader}, it can be shown that:
\begin{observation} \label{obs:constancy_with_invisible_reader}
	Consider $V \subseteq \V$, $k \in \N$, and a finite run $r$ with an invisible reader $p \in \Pi$.
	If $w \in \W$ is $\constant[k]{V}{\{p\}}$ in $r$ then $w$ is $\constant[k]{V}{\emptyset}$ in $r$.
\end{observation}

We next prove a stronger variant of Lemma~\ref{lem:burning_lemma} that allows us to add a permanent $\Write$ to the shared storage while some set $C \subseteq \W$ of $\Writes$ are constant.
Note that since the number of $\Writes$ of a value $v$ is infinite and the number of constant $\Writes$ in a finite run is finite, for any non-empty $V \subseteq \V$, $\W_V \setminus C$ is non-empty.

\begin{lemma} \label{lem:burning_lemma_writes}
	Consider a non-empty set of values $V \subseteq \V$, a set of readers $\Theta \subset \Pi$, a reader $p \in \Pi \setminus \Theta$, and a finite run $r$ of a wait-free regular \commonWrite\ algorithm.
	Let $C$ be a set of $\Writes$ that are $\constant{V}{\Theta}$ in $r$.
	Then there is an extension $r'$ of $r$ where some $w \in \W_V \setminus C$ returns and is $\permanent{\tau-1}{V}{\Theta \cup \{p\}}$, and s.t.\ in \suffix{r'}{r} $\Writes$ are limited to $\W_V$ and readers in $\Theta$ do not take actions.
\end{lemma}

\begin{proof}
	Assume by contradiction that the lemma does not hold.
	We build an extension $r'$ of $r$ where a $\Read_p$ operation includes infinitely many actions of $p$ yet does not return.
	To this end, we show that the following property holds at specific times in \suffix{r'}{r}:
	\[ \invariantWrites{\hat{r},t} \triangleq \forall w \in \writes{\hat{r},t} \cap \W_V \, : \, \abs{\alllabels{w,\hat{r},t}} < \tau. \]
	Note that, by definitions of \commonWrite{} and of $\alllabelsnoarg$, whenever $\invariantWrites{r',t}$ holds, no pending $\Read_p$ invocation can return a value $v \in \values{r',t} \cap V$.
	
	First, extend $r$ to $r_0$ by returning any pending $\Read_p$ and $\Write$, invoking and returning a $\wr{v_0}$ for some $v_0 \in \V$ (the operations eventually return, by wait-freedom), and finally invoking a $\Read_p$ operation $rd$ without allowing it to take actions.
	We now prove by induction that for all $k \in \N$, there exists an extension $r'$ of $r_0$ where (1) $\invariantWrites{r',\finalof{r'}}$ holds, (2) no $\Write$ is pending at $\finalof{r'}$, and in \suffix{r'}{r}: (3) $\Writes$ are restricted to $\W_V$, (4) $p$ performs $k$ actions on objects after invoking $rd$, (5) $rd$'s return is not enabled, and (6) processes in $\Theta$ do not take steps.
	
	Base: for $k = 0$, consider $r' = r_0$.
	(2,4,6) hold trivially.
	(3) holds since the only $\Write$ in \suffix{r'}{r} is $w_0 \in \W_V$.
	Since $p$ performs no actions following the invocation of $rd$, $\pdata{r',\finalof{r'}}$ is empty.
	Therefore, (1) $\invariantWrites{r',\finalof{r'}}$ is vacuously true, and $\llabels{w,r',\finalof{r'}}$ is empty for all $w \in \W_V$, thus (5) $rd$'s return is not enabled by \commonWrite{}.
	
	Step: assume inductively such an extension $r_1$ of $r_0$ with $k \ge 0$ actions by $p$ following $rd$'s invocation.
	Since $rd$ cannot return, by wait-freedom, an action $\step$ is enabled on some \object{}.
	We construct an extension $r_2$ of $r_1$ by letting $\step$ occur at time $\finalof{r_1}$.
	We then consider three cases:
	
	1. $p$ does not obtain a block with an origin $\Write$ in $\W_V \setminus \writes{\rfinal{r_1}}$ at $\step$, and therefore $\left( \writes{\rfinal{r_2}} \cap \W_V \right) \subseteq \left( \writes{\rfinal{r_1}} \cap \W_V \right)$.
	Then, by Observation~\ref{obs:no_new_labels} and the inductive hypothesis, (1) $\invariantWrites{\rfinal{r_2}}$ holds and thus, by \commonWrite{}, $rd$ cannot return any value $v \in \values{\rfinal{r_2}} \cap V$ at $\finalof{r_2}$.
	(5) It cannot return any other value in $\values{\rfinal{r_2}}$ by regularity, and $r_2$ satisfies the induction hypothesis for $k+1$ as (2,3,4,6) trivially hold.
	
	2. $p$ obtains a block with origin $\Write$ $w' \in C \cap \W_V \setminus \writes{\rfinal{r_1}}$ at $\step$.
	Then, in particular, $\abs{\llabels{w',\rfinal{r_1}}} = 0$. Since $w'$ is $\constant{V}{\Theta}$ in $r$ and in \suffix{r_1}{r} $\Writes$ are restricted to $\W_V$ and processes in $\Theta$ do not take steps (inductively), then by definition of constancy, $\abs{\slabels{w',\rfinal{r_1}}} = \tau - 1$.
	By Observation~\ref{obs:no_new_labels}, for all $w \in \writes{\rfinal{r_2}} \cap \W_V$: $\alllabels{w, \rfinal{r_2}} \subseteq \alllabels{w, \rfinal{r_1}}$.
	Thus, $\abs{\alllabels{w', \rfinal{r_2}}} \le \abs{\llabels{w',\rfinal{r_1}}} + \abs{\slabels{w',\rfinal{r_1}}} = \tau - 1$.
	Together with the inductive hypothesis, $\forall w \in \writes{\rfinal{r_2}} \cap \W_V \setminus \{w'\}$, $\abs{\alllabels{w, \rfinal{r_2}}} \le \abs{\alllabels{w, \rfinal{r_1}}} < \tau$; $\invariantWrites{r_2,\finalof{r_2}}$ follows, thus (5) follows, and (2,3,4,6) trivially hold.
	
	3. $p$ obtains a block with origin $\Write$ $w' \in \W_V \setminus \left( \writes{\rfinal{r_1}} \cup C \right)$ at $\step$.
	Then, in particular, $\abs{\llabels{w',r_2,\finalof{r_2}}} = 1$ and the number of labels of other $\Writes$ in $\writes{\rfinal{r_2}}$ does not increase following $\step$, thus $rd$'s return is not enabled at $\finalof{r_2}$ by \commonWrite\ and regularity.
	
	By the contradicting assumption, $w'$ is not $\permanent{\tau-1}{V}{\Theta \cup \{p\}}$ in $r_2$, thus there is an extension $r_3$ of $r_2$ s.t.\ $\abs{\slabels{w',\rfinal{r_3}}} < \tau-1$ and in \suffix{r_3}{r_2} $\Writes$ are limited to $\W_V$ and no readers in $\Theta \cup \{p\}$ take steps (3,4,6 hold).S
	We further extend $r_3$ to $r_4$ by letting any pending $\Write$ return (2).
	
	Let $S = \writes{\rfinal{r_2}} \cap \W_V$.
	Since every $w \in S$ returns before $\finalof{r_2}$ by the inductive assumption, the $\Writes$ in \suffix{r_4}{r_2} do not produce new labels associated with $w$.
	Since readers do not affect the sets $\slabelsnoarg$, it follows that $\forall w \in S\, : \, \slabels{w,\rfinal{r_4}} \subseteq \slabels{w,\rfinal{r_3}} \subseteq \slabels{w,\rfinal{r_2}}$.
	Next, $p$ takes no steps in \suffix{r_4}{r_2} (4 holds), thus $\forall w \in S: \, \llabels{w',\rfinal{r_4}} = \llabels{w',\rfinal{r_2}}$.
	It follows that:
	\begin{equation} \label{eq:psi_holds_for_w'}
	\abs{\alllabels{w',\rfinal{r_4}}} \le \abs{\llabels{w',\rfinal{r_2}}} + \abs{\slabels{w',\rfinal{r_3}}} < 1 + (\tau-1) = \tau.
	\end{equation}
	Moreover, by Observation~\ref{obs:no_new_labels} and the inductive assumption that $\invariantWrites{\rfinal{r_1}}$ holds,
	\begin{equation} \label{eq:psi_holds_for_S_minus_w'}
	\forall w \in S \setminus \{w'\} \, : \, \abs{\alllabels{w,\rfinal{r_4}}} \le \abs{\alllabels{w,\rfinal{r_1}}} < \tau.
	\end{equation} 
	
	From Equations~\ref{eq:psi_holds_for_w'} and~\ref{eq:psi_holds_for_S_minus_w'}, and since $\writes{\rfinal{r_4}} \cap \W_V = \writes{\rfinal{r_2}} \cap \W_V = S$, we get (1) $\invariantWrites{\rfinal{r_4}}$.
	Since $rd'$s return is not enabled at $\finalof{r_2}$ and (4) it took no actions since, its return is not enabled anywhere in \suffix{r_4}{r_1} (5), and we are done.
	\qedhere
\end{proof}

\subsection{Invisible reads} \label{subsect:commonWriteInvisibleReads}

We prove the following theorem by constructing a run with an exponential number of $\tau$-permanent values.
The idea is to show that if there is a value in the domain for which there is no $\tau$-permanent $\Write$, then infinitely many $\Writes$ remain $(\tau-1)$-constant, which is of course impossible.

\begin{theorem} \label{thm:common_write_invisible_reads_bound}
	The storage cost of a regular \commonWrite{} wait-free SRSW register emulation where $\Reads$ are invisible is at least \resCommonWriteInvisible{} blocks.
\end{theorem}

\begin{lemma} \label{lem:common_write_someone_constant_invisible}
	Consider a non-empty set of values $V \subseteq \V$ and a finite run $r$.
	Let $C$ be a set of $\Writes$ that are $\constant{V}{\emptyset}$ in $r$.
	Then there exists an extension $r'$ of $r$ where $\Writes$ in \suffix{r'}{r} are limited to $\W_V$, and some $w \in \W_V \setminus C$ is either $\constant{V}{\emptyset}$ or $\permanent{\tau}{V}{\emptyset}$ in $r'$.
\end{lemma}

\begin{proof}
	Let $p \in \Pi$ be a reader.
	By Lemma~\ref{lem:burning_lemma_writes}, there is an extension $r'$ of $r$ where $\Writes$ in \suffix{r'}{r} are limited to $\W_V$ and some $w \in \W_V \setminus C$ returns and is $\permanent{\tau-1}{V}{\{p\}}$.
	By Claim~\ref{claim:permanence_with_invisible_reader}, if $w$ is $\permanent{\tau}{V}{\{p\}}$ in $r'$, then $w$ is $\permanent{\tau}{V}{\emptyset}$ in $r'$ and the lemma follows.
	Otherwise, there exists an extension $r''$ of $r'$ where in \suffix{r''}{r'} $\Writes$ are limited to $\W_V$ and $p$ takes no steps, and $\abs{\slabels{w,r'',\finalof{r''}}} < \tau$.
	Since $w$ is $\permanent{\tau-1}{V}{\{p\}}$ in $r'$, $\abs{\slabels{w,r'',\finalof{r''}}} = \tau - 1$.
	
	We show that $w$ is $\constant{V}{\emptyset}$ in $r''$.
	Consider an extension $r'''$ of $r''$ where $\Writes$ are limited to values from $V$ and $p$ takes no steps in \suffix{r'''}{r''}.
	Since $w$ has already returned by time $\finalof{r''}$, no new blocks with an origin $\Write$ of $w$ can be added to the shared storage in $r'''$ after $\finalof{r''}$.
	It follows that $\slabels{w,\rfinal{r'''}} \subseteq \slabels{w,\rfinal{r''}}$.
	However, since $w$ is $\permanent{\tau-1}{V}{\{p\}}$ in $r'$, and in \suffix{r'''}{r'} $\Writes$ are limited $\W_V$ and $p$ takes no steps, then $\abs{\slabels{w, \rfinal{r'''}}} \ge \tau - 1 = \abs{\slabels{w,\rfinal{r''}}}$, yielding that $\slabels{w,r''',\finalof{r'''}} = \slabels{w,r'', \finalof{r''}}$.
	Thus, $w$ is $\constant{V}{\{p\}}$ in $r''$.
	The lemma follows from Observation~\ref{obs:constancy_with_invisible_reader}.
\end{proof}

\begin{claim} \label{claim:some_write_leaves_tau}
	Consider a finite run $r$ and a non-empty $V \subseteq \V$. Then there is an extension $r'$ of $r$ s.t.\ $\Writes$ in \suffix{r'}{r} are limited to $\W_V$, and some $w \in \W_V$ is $\permanent{\tau}{V}{\emptyset}$ in $r'$.
\end{claim}

\begin{proof}
	Consider an algorithm with storage cost $n$, and let $m = \ceil{n/(\tau-1)} + 1$.
	Assume by contradiction that the claim does not hold.
	We get a contradiction by constructing $m + 1$ extensions of $r$; $r_0, ..., r_m$ with sets of $\Writes$ $C_0 \subset C_1 \subset \cdots \subset C_m \subseteq \W_v$ s.t.\ for all $0 \le k \le m$:
	\begin{enumerate}
		\item $\Writes$ in \suffix{r_k}{r} are limited to $\W_V$, and
		\item $C_k$ is a set of $k$ $\Writes$ that are $\constant{V}{\emptyset}$ in $r_k$.
	\end{enumerate}
	Note that in $r_m$, $\ceil{\frac{n}{\tau-1}}+1$ $\Writes$ are $\constant{V}{\emptyset}$, implying a storage cost greater than $n$ by Observation~\ref{obs:n_ge_slabels}, a contradiction.
	
	The construction is by induction.
	The base case vacuously holds for $r_0 = r$, $C_0 = \emptyset$.
	Assume inductively such $r_k$ and $C_k$ for $k < m$.
	By Lemma~\ref{lem:common_write_someone_constant_invisible} there exists an extension $r_{k+1}$ of $r_k$ where some $w \in \W_V \setminus C_k$ is either $\permanent{\tau}{V}{\emptyset}$ or $\constant{V}{\emptyset}$, and $\Writes$ in \suffix{r_{k+1}}{r_k} are limited to $\W_V$.
	Since all $\Writes$ in $C_k$ are $\constant{V}{\emptyset}$ in $r_k$ they are also $\constant{V}{\emptyset}$ in $r_{k+1}$.
	By the contracting assumption, $w$ is not $\permanent{\tau}{V}{\emptyset}$ in $r_{k+1}$ thus it is $\constant{V}{\emptyset}$ in the run.
	Let $C_{k+1} = C_k \cup \{w\}$, therefore $\abs{C_{k+1}} = k + 1 $ and all $\Writes$ in $C_{k+1}$ are $\constant{V}{\emptyset}$ in $r_{k+1}$, as needed.
\end{proof}

We are now ready to prove our lower bound of \resCommonWriteInvisible\ blocks:

\begin{proof}[Proof (Theorem~\ref{thm:common_write_invisible_reads_bound})]
	We show that there exist $2^D+1$ finite runs $r_0, r_1, \ldots, r_{2^D}$ and sets of values $V_0 \supset V_1 \supset ... \supset V_{2^D}$ and $U_0 \subset U_1 \subset ... \subset U_{2^D}$, such that for all $0 \le k \le 2^D$:
	\begin{enumerate}
		\item $\abs{V_k} = 2^D - k$, $\abs{U_k} = k$, $V_k \cap U_k = \emptyset$, and
		\item all elements of $U_k$ are $\permanent{\tau}{V_{k}}{\emptyset}$ in $r_k$.
	\end{enumerate}
	
	By induction.
	Base: $r_0$ is the empty run, $V_0 = \V$, $U_0 = \emptyset$.
	Assume inductively such $r_k$, $V_k$, and $U_k$ for $k < 2^D$, and construct $r_{k+1}$ as follows:
	first, because $\abs{V_k} = 2^D - k > 0$, by Claim~\ref{claim:some_write_leaves_tau} there is an extension $r_{k+1}$ of $r_k$ where $\Writes$ in \suffix{r_{k+1}}{r_k} are limited to $\W_{V_k}$ and some $w \in \W_{V_k}$ is $\permanent{\tau}{V_{k}}{\emptyset}$.
	
	Consider the value $v \in V_k$ written by $w$.
	By Observation~\ref{obs:write-permanence-implies-value-permanence}, $v$ is $\permanent{\tau}{V_{k}}{\emptyset}$ in $r_{k+1}$.
	Let $V_{k+1} = V_k \setminus \{v\}$, then $\abs{V_{k+1}} = \abs{V_k} - 1 = 2^D - (k+1)$.
	Further let $U_{k+1} = U_{k} \cup \{v\}$.
	Note that, because $V_k \cap U_k = \emptyset$, we get $v \notin U_k$ and hence $V_{k+1} \cap U_{k+1} = \emptyset$ and $\abs{U_{k+1}} = \abs{U_k} + 1 = k+1$.
	Since $V_k \supset V_{k+1}$, then $v$ is $\permanent{\tau}{V_{k+1}}{\emptyset}$.
	Additionally, $\Writes$ in \suffix{r_{k+1}}{r_k} are from $\W_{V_k}$, thus by the inductive assumption and Observation~\ref{obs:permanence_in_constructions}, values in $U_k$ are $\permanent{\tau}{V_{k+1}}{\emptyset}$ in $r_{k+1}$, and so all of $U_{k+1}$ are $\permanent{\tau}{V_{k+1}}{\emptyset}$ in $r_{k+1}$.
	
	Finally, $U_{2^D}$ holds $2^D$ values that are $\permanent{\tau}{\emptyset}{\emptyset}$ in $r_{2^D}$.
	By Observation~\ref{obs:n_ge_slabels}:
	\[ n \ge \resCommonWriteInvisible. \qedhere \]
\end{proof}

\subsection{Visible reads} \label{subsect:commonWriteVisibleReads}

To prove a lower bound on the cost of systems with visible $\Reads$, we create a similar construction, except that the number of extensions might be limited by the number of readers, $R$.
Instead, the bound depends on $\min \left(2^D - 1 \, , \, R \right)$:

\begin{theorem} \label{thm:common_write_visible_reads_bound}
	The storage cost of a regular \commonWrite{} wait-free MRSW register emulation is at least \resCommonWriteVisible{} blocks.
\end{theorem}

Let $N = \commonWriteN$.
We build a run with $N$ $(\tau-1)$-permanent values:

\begin{lemma} \label{lem:construction_common_write_visible_reads}
	There exist finite runs $r_0, r_1, \ldots, r_{N}$, sets of values $V_0 \supset V_1 \supset ... \supset V_{N}$ and $U_0 \subset U_1 \subset ... \subset U_{N}$, and sets of readers $\Theta_0 \subset \Theta_1 \subset ... \subset \Theta_N$, s.t.\ for all $0 \le k \le N$:
	\begin{enumerate}
		\item $\abs{V_k} = 2^D - k$, $\abs{U_k} = \abs{\Theta_k} = k$, $V_k \cap U_k = \emptyset$, and
		\item all elements of $U_k$ are $\permanent{\tau-1}{V_k}{\Theta_k}$ in $r_k$.
	\end{enumerate}
\end{lemma}

\begin{proof}
	By induction.
	Base: $r_0$ is the empty run, $V_0 = \V$, $\Theta_0 = U_0 = \emptyset$.	
	Assume inductively such $r_k$, $V_k$, $U_k$, and $\Theta_k$ for $k < N$, and construct $r_{k+1}$ as follows:
	since $R - \abs{\Theta_k} > 0$, there is a reader $p \in \Pi \setminus \Theta_k$.
	Moreover, $\abs{V_k} > 2^D - N > 0$.
	Therefore, by Lemma~\ref{lem:burning_lemma_writes}, there is an extension $r_{k+1}$ of $r_k$ where $\Writes$ in \suffix{r_{k+1}}{r_k} are limited to $\W_{V_k}$, readers in $\Theta_k$ do not take steps in \suffix{r_{k+1}}{r_k}, and some $w \in W_{V_k}$ returns and is $\permanent{\tau-1}{V_k}{\Theta_k \cup \{p\}}$ in $r_{k+1}$.
	
	Let $\Theta_{k+1} = \Theta_k \cup \{p\}$, and consider the value $v \in V_k$ written by $w$.
	By Observation~\ref{obs:write-permanence-implies-value-permanence}, $v$ is $\permanent{\tau-1}{V_k}{\Theta_{k+1}}$.
	Let $V_{k+1} = V_k \setminus \{v\}$, then $\abs{V_{k+1}} = 2^D - (k+1)$.
	Further let $U_{k+1} = U_{k} \cup \{v\}$.
	Since $V_k \cap U_k = \emptyset$, we get that $V_{k+1} \cap U_{k+1} = \emptyset$ and $\abs{U_{k+1}} = k+1$.
	
	Since $V_k \supset V_{k+1}$, $v$ is $\permanent{\tau-1}{V_{k+1}}{\Theta_{k+1}}$.
	In addition, in \suffix{r_{k+1}}{r_k} $\Writes$ are limited to $\W_{V_k}$ and readers in $\Theta_k$ do not take steps, and since $\Theta_k \subset \Theta_{k+1}$, then by the inductive assumption and Observation~\ref{obs:permanence_in_constructions}, all values in $U_k$ are $\permanent{\tau-1}{V_{k+1}}{\Theta_{k+1}}$.
	Therefore all elements of $U_{k+1}$ are $\permanent{\tau-1}{V_{k+1}}{\Theta_{k+1}}$ in $r_{k+1}$, as needed.
\end{proof}

From Lemma~\ref{lem:construction_common_write_visible_reads}, in $r_N$ there is a set of $N$ $(\tau-1)$-permanent values, inducing a cost of $(\tau-1)N$.
We use Claim~\ref{claim:2tauMinusOne} to increase the bound by $2\tau-1$ additional blocks.

\begin{proof}[Proof (Theorem~\ref{thm:common_write_visible_reads_bound})]
	Construct $r_N$, $V_N$, $U_N$, and $\Theta_N$ as in Lemma~\ref{lem:construction_common_write_visible_reads}.
	Note that, since $R - N \ge 1$, there is a reader $p \in \Pi \setminus \Theta_N$.
	Since $V_N \cap U_N = \emptyset$ and $\abs{V_N} = 2^D - N = 2^D - (\commonWriteN) \ge 2$, 	$V_N$ contains two values, and they are not in $U_N$.
	Extend $r_N$ to $r_{N+1}$ by invoking and returning $\wr{v}$ and $\wr{v'}$ for $v,v' \in V_{N}$.
	
	By Claim~\ref{claim:2tauMinusOne}, there is a time $t \ge \finalof{r_N}$ in $r_{N+1}$ when there are $2\tau-1$ blocks in the \sharedStorage{} with origin values of $v$ or $v'$.
	In addition, $U_N$ consists of $N$ values that are $\permanent{\tau-1}{V_N}{\Theta_N}$ in $r_N$, and since in \suffix{r_{N+1}}{r_N} $\Writes$ are of values from $V_N$ and no reader in $\Theta_N$ takes steps, the values in $U_N$ remain $\permanent{\tau-1}{V_N}{\Theta_N}$ in $r_{N+1}$.
	By Observation~\ref{obs:n_ge_slabels}, the storage cost amounts to at least:
	\[ n \ge 2\tau - 1 + (\tau-1) N	= \tau + (\tau-1)(N+1) = \resCommonWriteVisible.
	\qedhere
	\]
\end{proof}

\section{Discussion}
\label{sec:discussion}

We have shown lower bounds on the space complexity of regular wait-free \disint{} algorithms.
Although our bounds are stated in terms of blocks, there are scenarios where they entail concrete bounds in terms of bits.
In replication, each block stores an entire value, thus the block sizes are $D$ bits.
Other applications use symmetric coding where all blocks are of equal size.
Using a simple pigeonhole argument, it can be shown that in \disint{} emulations that use symmetric coding and that are not $(\tau+1)$-disintegrated, the size of blocks is at least $D/\tau$ bits, yielding bounds of $D \cdot 2^D$ and $D + D \frac{\tau-1}{\tau} \cdot \min \left( 2^D - 1 \, , \, R \right)$ with invisible and visible readers, respectively.

Our lower bounds for the common write case explain, for the first time, why previous coded storage algorithms have either had the readers write or consumed exponential (or even unbounded) space.
Similarly, they establish why previous emulations of large registers from smaller ones have either had the readers write, had the writer share blocks among different $\Writes$, or consumed exponential space.

Our work leaves several open questions.
First, when replication is used as a means to overcome Byzantine faults or data corruption, our results suggest that there might be an interesting trade-off between the shared storage cost and the size of local memory at the readers, and a possible advantage to systems that apply replication rather than error correction codes:
we have shown that, with invisible readers, the former require $\Omega(2^D/L)$ blocks, rather than the $\Omega(2^D)$ blocks needed by the latter.
Whether there are algorithms that achieve this lower cost remains an open question.
Second, it is unclear how the bounds would be affected by removing our assumption that each block in the shared storage pertains to a single write.
Wei~\cite{wei2018space} has provided a partial answer to this questions by showing that similar bounds hold without this assumption, but only in the case of emulating large registers from smaller ones \emph{without} \meta{} at all.
Similarly, it would be interesting to study whether allowing readers to write data (and not only signals) impacts the storage cost.
Finally, future work may consider additional sub-classes of disintegrated storage, \eg with unresponsive objects, and show that additional costs are incurred in these cases.

\section*{Acknowledgments}
\label{sec:acknowledgments}

We thank Yuval Cassuto, Gregory Chockler, Rati Gelashvili, and Yuanhao Wei for many insightful discussions on space bounds for coded storage and emulations of large registers from smaller ones.

\bibliography{bibliography}

\begin{thebibliography}{10}

\bibitem{faultyMemorybyzantineDiscPaxos}
Ittai Abraham, Gregory Chockler, Idit Keidar, and Dahlia Malkhi.
\newblock Byzantine disk paxos: optimal resilience with byzantine shared
  memory.
\newblock {\em Distributed Computing}, 18(5):387--408, 2006.

\bibitem{abraham2007wait}
Ittai Abraham, Gregory Chockler, Idit Keidar, and Dahlia Malkhi.
\newblock Wait-free regular storage from byzantine components.
\newblock {\em Information Processing Letters}, 101(2):60--65, 2007.

\bibitem{aguileraDisks2003podc}
Marcos~K. Aguilera, Burkhard Englert, and Eli Gafni.
\newblock On using network attached disks as shared memory.
\newblock In {\em Proceedings of the Twenty-second Annual Symposium on
  Principles of Distributed Computing}, PODC '03, pages 315--324, New York, NY,
  USA, 2003. ACM.
\newblock URL: \url{http://doi.acm.org/10.1145/872035.872082}, \href
  {http://dx.doi.org/10.1145/872035.872082} {\path{doi:10.1145/872035.872082}}.

\bibitem{codingAguilera2005using}
Marcos~Kawazoe Aguilera, Ramaprabhu Janakiraman, and Lihao Xu.
\newblock Using erasure codes efficiently for storage in a distributed system.
\newblock In {\em 2005 International Conference on Dependable Systems and
  Networks (DSN'05)}, pages 336--345, June 2005.

\bibitem{androulaki2014erasure}
Elli Androulaki, Christian Cachin, Dan Dobre, and Marko Vukoli{\'c}.
\newblock Erasure-coded byzantine storage with separate metadata.
\newblock In {\em International Conference on Principles of Distributed
  Systems}, pages 76--90. Springer, 2014.

\bibitem{ABD}
Hagit Attiya, Amotz Bar-Noy, and Danny Dolev.
\newblock Sharing memory robustly in message-passing systems.
\newblock {\em Journal of the ACM (JACM)}, 42(1):124--142, January 1995.
\newblock URL: \url{http://doi.acm.org/10.1145/200836.200869}, \href
  {http://dx.doi.org/10.1145/200836.200869} {\path{doi:10.1145/200836.200869}}.

\bibitem{bazzi2004non}
Rida~A Bazzi and Yin Ding.
\newblock Non-skipping timestamps for byzantine data storage systems.
\newblock In {\em International Symposium on Distributed Computing}, pages
  405--419. Springer, 2004.

\bibitem{codingCachin2006optimal}
Christian Cachin and Stefano Tessaro.
\newblock Optimal resilience for erasure-coded byzantine distributed storage.
\newblock In {\em Dependable Systems and Networks, 2006. DSN 2006.
  International Conference on}, pages 115--124. IEEE, 2006.

\bibitem{codingLynchCadambe2014coded}
Viveck~R. Cadambe, Nancy Lynch, Muriel Medard, and Peter Musial.
\newblock A coded shared atomic memory algorithm for message passing
  architectures.
\newblock In {\em Network Computing and Applications (NCA), 2014 IEEE 13th
  International Symposium on}, pages 253--260. IEEE, 2014.

\bibitem{cadambe2016podc}
Viveck~R. Cadambe, Zhiying Wang, and Nancy Lynch.
\newblock Information-theoretic lower bounds on the storage cost of shared
  memory emulation.
\newblock In {\em Proceedings of the 2016 ACM Symposium on Principles of
  Distributed Computing}, PODC '16, pages 305--313, New York, NY, USA, 2016.
  ACM.
\newblock URL: \url{http://doi.acm.org/10.1145/2933057.2933118}, \href
  {http://dx.doi.org/10.1145/2933057.2933118}
  {\path{doi:10.1145/2933057.2933118}}.

\bibitem{chaudhuri2000one}
Soma Chaudhuri, Martha~J Kosa, and Jennifer~L Welch.
\newblock One-write algorithms for multivalued regular and atomic registers.
\newblock {\em Acta Informatica}, 37(3):161--192, 2000.

\bibitem{chen2016opodis}
Tian~Ze Chen and Yuanhao Wei.
\newblock {Step Optimal Implementations of Large Single-Writer Registers}.
\newblock In Panagiota Fatourou, Ernesto Jim{\'e}nez, and Fernando Pedone,
  editors, {\em 20th International Conference on Principles of Distributed
  Systems (OPODIS 2016)}, volume~70 of {\em Leibniz International Proceedings
  in Informatics (LIPIcs)}, pages 32:1--32:16, Dagstuhl, Germany, 2017. Schloss
  Dagstuhl--Leibniz-Zentrum fuer Informatik.
\newblock URL: \url{http://drops.dagstuhl.de/opus/volltexte/2017/7101}, \href
  {http://dx.doi.org/10.4230/LIPIcs.OPODIS.2016.32}
  {\path{doi:10.4230/LIPIcs.OPODIS.2016.32}}.

\bibitem{chockler2007amnesic}
Gregory Chockler, Rachid Guerraoui, and Idit Keidar.
\newblock Amnesic distributed storage.
\newblock In {\em Distributed Computing}, pages 139--151. Springer, 2007.

\bibitem{chocklerSpiegelmanPODC2017}
Gregory Chockler and Alexander Spiegelman.
\newblock Space complexity of fault-tolerant register emulations.
\newblock In {\em Proceedings of the ACM Symposium on Principles of Distributed
  Computing}, PODC '17, pages 83--92, New York, NY, USA, 2017. ACM.
\newblock URL: \url{http://doi.acm.org/10.1145/3087801.3087824}, \href
  {http://dx.doi.org/10.1145/3087801.3087824}
  {\path{doi:10.1145/3087801.3087824}}.

\bibitem{dobre2013powerstore}
Dan Dobre, Ghassan Karame, Wenting Li, Matthias Majuntke, Neeraj Suri, and
  Marko Vukoli{\'c}.
\newblock Powerstore: proofs of writing for efficient and robust storage.
\newblock In {\em Proceedings of the 2013 ACM SIGSAC conference on Computer \&
  communications security}, pages 285--298. ACM, 2013.

\bibitem{codingRashid}
Partha Dutta, Rachid Guerraoui, and Ron~R. Levy.
\newblock Optimistic erasure-coded distributed storage.
\newblock In {\em Proceedings of the 22nd International Symposium on
  Distributed Computing}, DISC '08, pages 182--196, Berlin, Heidelberg, 2008.
  Springer-Verlag.
\newblock URL: \url{http://dx.doi.org/10.1007/978-3-540-87779-0_13}, \href
  {http://dx.doi.org/10.1007/978-3-540-87779-0_13}
  {\path{doi:10.1007/978-3-540-87779-0_13}}.

\bibitem{codingGoodson2004efficient}
Garth~R Goodson, Jay~J Wylie, Gregory~R Ganger, and Michael~K Reiter.
\newblock Efficient byzantine-tolerant erasure-coded storage.
\newblock In {\em Dependable Systems and Networks, 2004 International
  Conference on}, pages 135--144. IEEE, 2004.

\bibitem{faultyMemoryJayanti1998}
Prasad Jayanti, Tushar~Deepak Chandra, and Sam Toueg.
\newblock Fault-tolerant wait-free shared objects.
\newblock {\em Journal of the ACM (JACM)}, 45(3):451--500, 1998.

\bibitem{lamportRegular}
Leslie Lamport.
\newblock On interprocess communication.
\newblock {\em Distributed computing}, 1(2):86--101, 1986.

\bibitem{martin2002minimal}
Jean-Philippe Martin, Lorenzo Alvisi, and Michael Dahlin.
\newblock Minimal byzantine storage.
\newblock In {\em International Symposium on Distributed Computing}, pages
  311--325. Springer, 2002.

\bibitem{peterson1983concurrent}
Gary~L Peterson.
\newblock Concurrent reading while writing.
\newblock {\em ACM Transactions on Programming Languages and Systems (TOPLAS)},
  5(1):46--55, 1983.

\bibitem{spaceBounds}
Alexander Spiegelman, Yuval Cassuto, Gregory Chockler, and Idit Keidar.
\newblock Space bounds for reliable storage: Fundamental limits of coding.
\newblock In {\em Proceedings of the 2016 ACM Symposium on Principles of
  Distributed Computing}, PODC '16, pages 249--258, New York, NY, USA, 2016.
  ACM.
\newblock URL: \url{http://doi.acm.org/10.1145/2933057.2933104}, \href
  {http://dx.doi.org/10.1145/2933057.2933104}
  {\path{doi:10.1145/2933057.2933104}}.

\bibitem{wangCadambeMultiversion2014}
Zhiying Wang and Viveck~R. Cadambe.
\newblock On multi-version coding for distributed storage.
\newblock In {\em Communication, Control, and Computing (Allerton), 2014 52nd
  Annual Allerton Conference on}, pages 569--575. IEEE, 2014.

\bibitem{wei2018space}
Yuanhao Wei.
\newblock Space complexity of implementing large shared registers.
\newblock {\em arXiv preprint arXiv:1808.00481}, 2018.

\end{thebibliography}
\end{document}